\documentclass[10pt,conference]{IEEEtran}
\usepackage{cite}
\usepackage{graphicx,color,epsfig,rotating}
\usepackage{amsfonts,amsmath,amssymb}
\usepackage{algorithm,algorithmic}
\usepackage{subfig}
\usepackage{url}

\long\def\comment#1{}

\newfont{\bbb}{msbm10 scaled 700}

\newfont{\bb}{msbm10 scaled 1100}

\newtheorem{definition}{Definition}
\newtheorem{theorem}{Theorem}

\newtheorem{proof}{Proof}

\begin{document}
										
\title{Graph Theory versus Minimum Rank\\ for Index Coding}
\author{
\IEEEauthorblockN{Karthikeyan Shanmugam \IEEEauthorrefmark{1}, Alexandros G. Dimakis \IEEEauthorrefmark{1} and Michael Langberg \IEEEauthorrefmark{2}}
 
\IEEEauthorblockA{\IEEEauthorrefmark{1}
Department of Electrical and Computer Engineering, University of Texas at Austin, Austin, TX 78712-1684. \\
\texttt{karthiksh@utexas.edu,dimakis@austin.utexas.edu}}

\IEEEauthorblockA{\IEEEauthorrefmark{2}
Department of Electrical Engineering, SUNY at Buffalo, Buffalo, NY 14260. \\
\texttt{mikel@buffalo.edu}}
}

\maketitle
\begin{abstract}
We obtain novel index coding schemes and show that they provably outperform all previously known graph theoretic bounds proposed so far.
Further, we establish a rather strong negative result: all known graph theoretic bounds are within a logarithmic factor from the chromatic number. This is in striking contrast to $\mathrm{minrank}$ since prior work has shown that it can outperform the chromatic number by a polynomial factor in some cases. The conclusion is that all known graph theoretic bounds are not much stronger than the chromatic number. 
\end{abstract}

\section{Introduction}

Index coding is a fundamental network information theory problem with deep connections with combinatorial optimization and graph theory~\cite{birk1998informed,bar2006index,alon2008broadcasting,blasiak2010index,Kim2013ISIT,unal2013}. Interest in index coding is further increasing due to two recent developments: 
The first is that it was recently shown~\cite{el2010index,effros2012equivalence} that any arbitrary 
network coding problem with potentially multiple sources and receivers can be mapped to a properly constructed index coding instance. Therefore, statements about index coding can be translated to constructions or bounds for general networks, showing the surprising expressiveness of the problem. 
Second, interference alignment alongside information theoretic approaches have been recently applied for index coding~\cite{maleki2012index,jafartopology,tahmasbi2013critical,unal2013,Kim2013ISIT} introducing new interesting techniques for code constructions. Briefly, index coding is a noiseless broadcast problem where $m$ messages needs to be sent to $n$ users each requesting one of the $m$ messages. In addition, every user has some \textit{side information packets} which is a subset of the $m$ messages not including the request. Index coding capacity refers to the minimum number of (coded) transmissions required to satisfy all users. When $m=n$ and user requests do not overlap the problem can be  represented in terms of a directed side information graph $G_d$. A directed edge $(i,j)$ means that user $i$ has packet requested by $j$.

Methods for constructing index codes (\textit{i.e.} upper bounds for index coding) can be broadly separated in two categories: graph theoretic methods and algebraic methods relying on rank minimization. The focus of this work is on the former. Graph theoretic methods start from the well-known fact that all the users forming a clique in the side information digraph can be simultaneously satisfied by transmitting the XOR of their packets~\cite{birk1998informed}. This idea shows that the number of cliques required to cover all the vertices of the graph (the \text{clique cover number}) is an achievable upper bound. It is easy to see that the chromatic number of the complement graph is equal to the clique cover number. This is because all the vertices assigned to the same color cannot share an edge and hence must form a clique on the complement graph. It turns out that the idea based on coloring lead to a family of stronger bounds, starting with an LP relaxation called fractional chromatic number~\cite{blasiak2010index} and the stronger local chromatic number~\cite{Localpap1} which can be further fractionalized. Instead of covering with cliques, one can cover the vertices with cycles and obtain cycle cover bounds~\cite{bar2006index}. Another achievable scheme called \textit{partition multicast} was proposed~\cite{tehrani2012bipartite} which generalized both cycle and clique covers. In partition multicast, one first partitions the graph into subgraphs corresponding to sub-problems of the given index coding problem before choosing an appropriate covering for each subgraph. 

The second family of bounds is algebraic and requires minimizing the rank over all matrices that respect the structure of the side information graph over a finite field. It turns out \cite{bar2006index} that (for a given field size), scalar linear index coding is equal to the $\mathrm{minrank}$ quantity introduced by 
Haemers~\cite{haemers1978upper} in 1978 to obtain a bound for the Shannon graph capacity~\cite{shannon1956zero}.
Therefore, $\mathrm{minrank}$ characterizes the best possible scalar linear index code for a given finite field. 
Throughout this paper, we refer to the former family of bounds as graph-theoretic and the latter as algebraic. 
%We call an index coding bound algebraic if it requires the use of fields or rank. 

The main question we investigate in this paper is how all these quantities compare. We introduce a new graph theoretic bound and show that it provably outperforms all previous graph bounds. Our bound is obtained by combining all previous graph theoretic ideas discussed above: local coloring, fractionalization and partitioning. 
We then prove a rather strong negative result: \textit{all known graph theoretic bounds are within a constant factor from the fractional chromatic number}. Previous work has established that the fractional chromatic number is within a $\log n$ factor from the chromatic number \cite{feige1998threshold}. Therefore, all known graph bounds can improve, at most, a $\log n$ factor from the chromatic number. This is in striking contrast to $\mathrm{minrank}$ where prior work has shown~\cite{lubetzky2009nonlinear,blasiak2010index} that it can outperform the chromatic number by a polynomial factor.  We emphasize that this performance benefit of $\mathrm{minrank}$ is shown only for special graph constructions~\cite{lubetzky2009nonlinear} and there are other examples 
where the fractional chromatic number can outperform $\mathrm{minrank}$. 
%The conclusion we obtain from our results is that all known graph theoretic bounds are not much stronger than the chromatic number. 

Depending on the structure of the side information graph, index coding can be investigated for undirected
(\textit{i.e.} symmetric side information) or, more generally directed graphs. In even greater generality, if we allow multiple users to request the same packet we can describe the problem with a hypergraph or with a bipartite directed graph~\cite{blasiak2010index,jafartopology,tehrani2012bipartite}. We refer to directed graph problems as unicast index coding (UIC) and more general hypergraphs as groupcast index coding (GIC).
   
We summarize our results and the previously known relationships between graph parameters in Figure \ref{Fig:Resultdiag}. We present the results for both the groupcast index coding (GIC) and the unicast index coding (UIC) scenarios. It should be noted that GIC results are more general and the directed graph parameters are included for readability. The blue box in the figure indicates previously known known parameters and relationships. Formal definitions will be given in Section \ref{Sec:Defn}.

\begin{figure}   
\centering
\includegraphics [width=9cm]{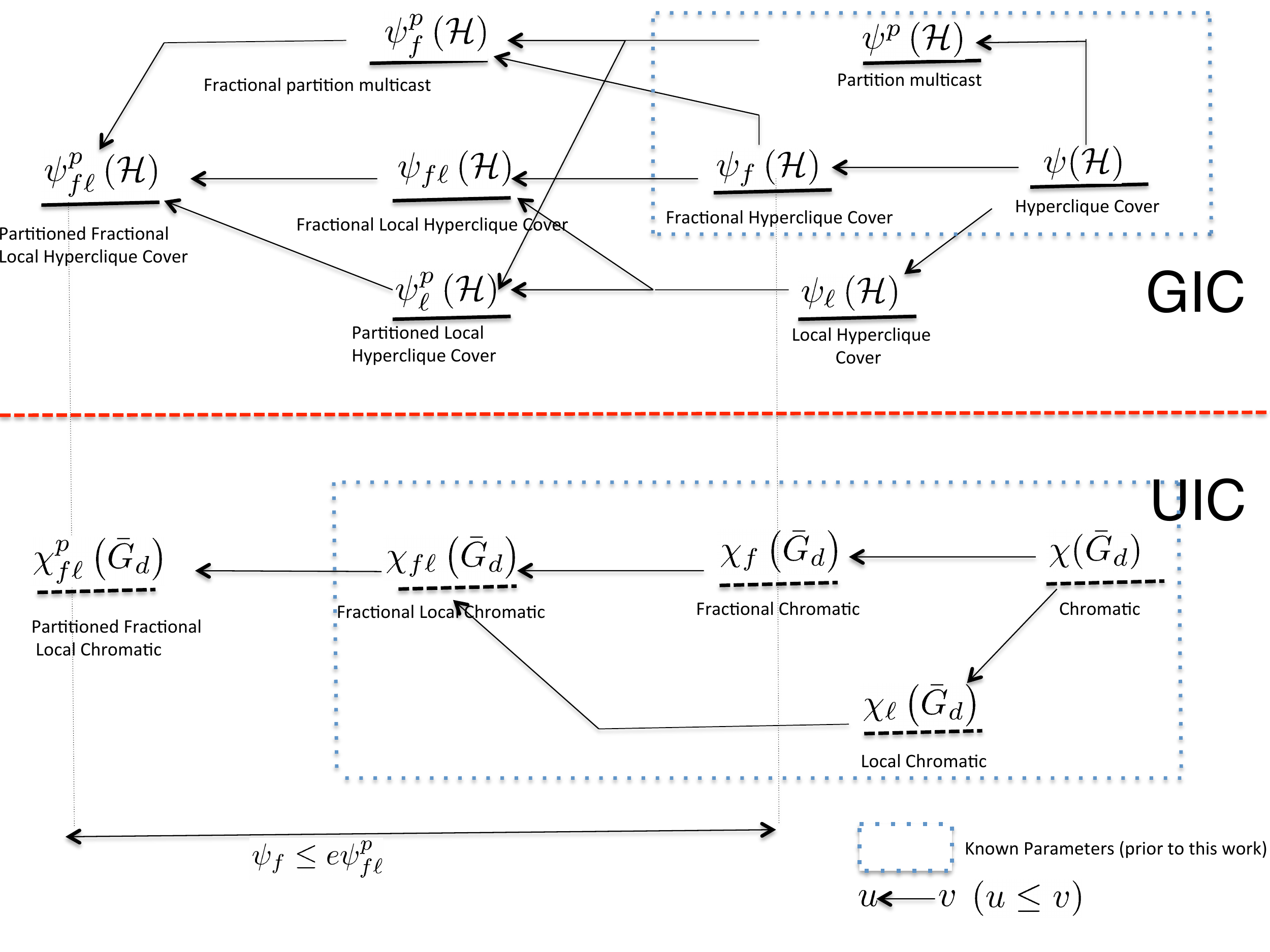}
\caption{\footnotesize Summary of our contributions. The bottom part of the figure describes index coding bounds for directed graphs (Unicast Index coding (UIC)) while the top describes the more general case of hypergraphs (Groupcast index coding (GIC)). Smaller graph quantities are placed to the left and the weakest bound (chromatic number) is placed to the rightmost of the figure. Arrows indicate the relationship they satisfy. An important result is shown at the bottom of the figure, illustrating that the 
best bounds we obtain ($\psi_{f \ell}^{p}$ and $\chi_{f \ell}^{p}$ are within a factor of $e$ from the fractional chromatic number $\chi_f$ and its hypergraph generalization $\psi_f$ respectively.)  }   
\label{Fig:Resultdiag}
\end{figure}

We outline our contributions as follows:
        
    \subsection{Our Contributions:}  
      \begin{enumerate}
         \item We start by extending the work of \cite{Localpap1}, that addressed the concept of local chromatic number in the context of unicast index coding to the more general groupcast setting. We define new parameters called \textit{local hyperclique cover} and its fractional version. These are the group cast analogues to local and fractional local chromatic numbers. We show that these have index coding achievable schemes. Further, we show that these parameters are 
within a factor of $e$ away from the fractional hyperclique cover. This is the natural generalization of 
the fractional chromatic number for the groupcast case.
         \item We define another parameter, called \textit{partitioned local hyperclique cover} and its fractional version for the groupcast setting. 
We show that this scheme is stronger than the ones based on local hyperclique cover and partition multicast and therefore all known graph-theoretic bounds. This parameter combines the ideas behind local coloring and partition multicast to provide a better index coding scheme.
         \item Finally, for our negative result, we show that this new scheme is within a factor $e$ from the fractional hyperclique cover (implying the same for all previous bounds as well).
      \end{enumerate}  
      
In the two subsequent sections, we provide detailed definitions for all the quantities used in this paper including the novel ones. Subsequently we state our results that bound the relationships of new quantities and also some unknown relationships of previously introduced quantities. Due to space constraints we omit most proofs that can be found in the long version of this manuscript \cite{Urllocal}.

%-------------

\section{Definitions and review of existing parameters}\label{Sec:Defn}
  For ease of notation, let $[n]$ denote the set $\{1,2 \ldots n \}$. $A-B$ is the set difference between sets $A$ and $B$. Let $G_d(V,E_d)$ be a directed graph on $n$ vertices. If $u \in V$. Let $N(u)$ denote the directed out-neighborhood, i.e. $N(u)= \{v \in V: (u,v) \in E_d \}$. Let $\left( N(u) \right)^c = V-N(u)-u$. Let $\bar{G}_d \left( V, \bar{E}_d \right)$ denote the \textit{directed complement} of $G_d$ which is another directed graph where out-neighborhood of vertex $u$ is $\left( N(u) \right)^c$. Let $2^A$ be the power set of $A$.   We define a groupcast index coding problem input instance using a directed bipartite graph as follows. 
  
   \begin{definition}
                      A Groupcast Index Coding problem (GIC) instance is given by the set $\{U,P, {\cal H} (U,P,L)\}$. $U=\{1,2,3 \ldots n\}$ is the set of users with $\lvert U \rvert=n$, $P=\{x_1,x_2 \ldots x_m\}$ is the set of packets with $\lvert P \rvert =m,~n \geq m$. ${\cal H}$ is a directed bipartite graph between the sets $U$ and $P$ with $L$ as the set of directed edges. Each packet $x_i \in \Sigma$ where $\Sigma$ is some alphabet. Every user $u$ requests a single packet $R(u) \in P$ and it has $S(u) \subset P-R(u)$ as \textit{side information}. If the request of user $u$ is $R(u)=p$, then the directed edge $(u,p) \in L$. If $p \in S(u)$, the directed edge $(p,u) \in L$.  $\hfill \lozenge$
   \end{definition}                   
                      
              Another representation of GIC is in terms of a directed hypergraph \cite{alon2008broadcasting}. In this representation, the problem is represented as a directed hypergraph ${\cal G} \left( P, U\right)$ such that $P$ is the set of vertices and every user $u$ corresponds to a directed hyperedge $\left(R(u), S(u) \right)$. In this work, we use adopt the equivalent directed bipartite graph representation of \cite{tehrani2012bipartite}.
                     
                     We assume w.l.o.g. that for all $u$, $\lvert R(u) \rvert=1$. Let $\left( S(u) \right)^c = P - S(u)-u$. Let $W(p)$ denote the set of all users who want packet $p$. $R(A)=\bigcup \limits_{u \in A}R(u)$ and $W(P)=\bigcup \limits_{p \in P}W(p)$. Note that a packet can be requested by multiple users.
                     
                      The GIC problem involves a common broadcasting agent who needs to satisfy all user requests with a minimum number of bits over a public broadcast noiseless channel. The agent is cognitive of all the side information present at every user. Transmitted bits are decoded at each user using its side information to recover the desired packets.  
In what follows, we define the minimum broadcast rate for the GIC problem. We define a valid index code for the GIC problem as follows:
 
         \begin{definition}
			(\textit{Valid index code})  Here, for notational reasons, assume $R(u)=x_i \in \Sigma$  is the packet desired by user $i$. A \textit{ valid index code} over the alphabet $\Sigma$ is a set $(\phi,\{\gamma_i\})$ consisting of:
  \begin{enumerate}
        \item An encoding function $\phi:\Sigma^m \rightarrow \{ 0,1\}^p$ which maps the $m$ packets to a transmitted message of length $p$ bits for some integral $p$. 
        \item   $n$ decoding functions $\gamma_i$ such that for every user $i$, $\gamma_u(\phi \left( x_1,x_2 \ldots x_m \right),S(i)) = x_i$ for all $\left[ x_1 ~x_2 \ldots x_n\right] \in \Sigma^m$. In other words, every user would be able to decode its desired message from the transmitted message and the side information.  $\hfill \lozenge$ 
  \end{enumerate}
           
\end{definition}

 The \textit{broadcast rate} $\beta_{\Sigma}({\cal H},\phi, \{ \gamma_i \})$ of the $(\phi,\{ \gamma_i \})$ index code for the GIC on ${\cal H}$ is the number of transmitted bits per received message bit at every user, i.e. $\beta_{\Sigma}(G_d,\phi,\{ \gamma_i\})= \frac{p}{\log_2 \lvert \Sigma \rvert}$.
        \begin{definition}  
             (\textit{Minimum broadcast rate})  The minimum broadcast rate $\beta({\cal H})$ is the minimum possible broadcast rate of all valid index codes over all alphabets $\Sigma$, i.e. $\beta({\cal H})= \inf \limits_{\Sigma} \inf \limits_{\phi,\{\gamma_i \}} \beta_{\Sigma} ({\cal H},\phi,\{ \gamma_i \})$. 
         $\hfill \lozenge$ 
 \end{definition}
   
       Now, we digress slightly by discussing an important special case of the GIC problem. A unicast index coding problem (UIC) is a special case of GIC where user requests do not overlap. Hence, without loss of generality, we take $m=n$ and take $P=U$ (packets and users are indistinguishable and user $i$ requests packet $x_i$). Therefore, one can represent a UIC problem using a directed side information graph $G_d$ with vertex set $U$ where the out-neighborhood of user $u$ is $N(u)=S(u)$. 
   \begin{definition}
              (\textit{Interference graph}) The interference graph, denoted by $\bar{G}_d(V,\bar{E}_d)$ of an UIC problem is a \textit{ directed complement} $\bar{G}_d$ of the side information graph $G_d$.      $\hfill \lozenge$ 
\end{definition}   
         
%         \begin{definition}
%           (\textit{Underlying undirected side information graph})The \textit{underlying undirected side information graph} of $G_d$, denoted by $G_u(V,E_u)$, is the graph obtained by deleting uni-directed edges (i.e. $(i,j) \in E_d$ but $(j,i) \notin E_d$) and all remaining bi-directed edges are replaced by an undirected edge, denoted by $\{i,j\}$.
%      $\hfill \lozenge$
%       \end{definition}
%         
%               We observe that the complement of the underlying undirected side information graph $G_u$ is the graph denoted by $\bar{G}_u(V,\bar{E}_u)$ which can alternatively be obtained by ignoring the orientation of the edges in the interference graph $\bar{G}_d$ (here, bi-directed edges in $\bar{G}_d$ can be replaced by a single undirected edge in $\bar{G}_u$). The various graphs associated with the index coding problem and relationships between them are illustrated in Fig. \ref{Graphdrawind}. The graphs in Fig. \ref{Graphdrawind} correspond to the index coding problem defined in Fig. \ref{Graphdrawind1}
	 
	   We now present a number of previously studied upper bounds on $\beta({\cal H})$ for GIC. The first is a bound from \cite{blasiak2010index}, referred to as the \textit{fractional hyperclique cover} and denoted here by $\psi_f \left( {\cal H} \right)$. Our definition below slightly differs from that in \cite{blasiak2010index} but nevertheless is equivalent. 
	   
	  \begin{definition}	
             (\textit{Weak Hyperclique}) A weak hyper clique $C \subseteq U$ is such that for any pair $u,v \in C$, we have $ ( u \in S(v) ~\mathrm{AND}~ v \in S(u) )~ \mathrm{OR}~ R(u) =R(v)$.          $\hfill \lozenge$
        \end{definition}
        
      Observe that in the GIC problem, one can satisfy all the users in $C$ by XORing their requests $R(C)$. This implies that a ``cover'' of the hypergraph by weak hypercliques implies a corresponding valid index code. In the rest of the paper, we use the term ``hyperclique'' instead of ``weak hyperclique''.
      
     \begin{definition}
                 The \textit{hyperclique cover} of ${\cal H}$, denoted by $\psi({\cal H})$, is given by the following Integer Program:
                       \begin{align}\label{prog:weakhyp}
                                     & \min \sum \limits_{C \in {\cal C}} y_C \nonumber\\
                                      \mathrm{s.t.} & \sum \limits_{C: u \in C} y_C = 1, ~\forall u \in U \nonumber\\
                                        \hfill & y_C \in \{0,1\}, ~\forall C \in {\cal C}
                       \end {align}
    where ${\cal C}$ is the set of all hypercliques in ${\cal H}$.
         $\hfill \lozenge$
 \end{definition}
          The LP relaxation of $(\ref{prog:weakhyp})$ is the fractional hyperclique cover $\psi_f \left( {\cal H} \right)$.   Now, provide some intuition behind program (\ref{prog:weakhyp}). A feasible solution to  (\ref{prog:weakhyp})  is a set of chosen hypercliques such that every user is covered exactly by one hyperclique. The least number of hypercliques required to cover every user is given by $\psi$. This implies that $\beta \leq \psi({\cal H})$ by our discussion above. In the UIC problem, a hyperclique is equivalent to a clique on $G_d$ (a clique in a directed graph is a complete subgraph where there are edges in both directions between any two vertices). Therefore, the fractional chromatic number, defined on the directed complement $\bar{G}_d$, is the equivalent of $\psi_f$. It is denoted by $\chi_f \left( \bar{G}_d \right)$.

  We now turn to discuss an additional scheme for GIC, partition multicast,  introduced in \cite{tehrani2012bipartite}. The scheme is a generalization of both cycle cover and  hyperclique cover. Formal definition is given below:
        
     \begin{definition}
        The \textit{partition multicast number} of ${\cal H}$, denoted $\psi^p \left( {\cal H} \right)$, is given by the following integer program:
                       \begin{align}
                                     \hfill & \min~ \sum \limits_{M} a_M d_M \nonumber \\
                                      \mathrm{s.t.}~ & \sum \limits_{M: v \in M} a_M = 1, ~\forall u \in U \nonumber\\
                                        \hfill & a_M \in \{0,1\}, ~\forall M \in 2^{U}-\{\emptyset \}                     
                                        \label{pr:partmul}                
                       \end {align}   
              where ${\cal C}$ is the set of hypercliques in ${\cal H}$ and $d_M= \lvert R(M) \rvert- \min \limits_{u \in M} \lvert R(M) \bigcap S(u) \rvert$. $\hfill \lozenge$  
     \end{definition}  
  
We provide some intuition behind (\ref{pr:partmul}). A feasible solution chooses a family of subsets of users (based on the value of $a_M$). We call each subset a multicast group. Every user is covered by exactly one such group. The bipartite subgraph, induced by a multicast group $M$ and packets demanded by $M$ is denoted ${\cal H} \left( M,R(M) \right)$. Every user has at least $\min \limits_{u \in M} \lvert R(M) \bigcap S(u) \rvert$ packets from $R(M)$. It was shown in \cite{tehrani2012bipartite} that $d_M$ coded transmissions using an $(\lvert R(M) \rvert, d_M  )$ MDS code allows users in group $M$ to recover their packet. The program (\ref{pr:partmul}) partitions the user set into an optimum set of multicast groups depending on the cost ($d_M$) of transmission for each group.

 \section{Definitions for New parameters}
   In this section, we provide definitions of new parameters that will be shown to have achievable index coding schemes for the GIC problem. We begin with the definition for a fractional version of the partition multicast scheme.
   
    \begin{definition}
          The \textit{fractional partition multicast number} of ${\cal H}$ , denoted $\psi_{f}^p \left( {\cal H} \right)$, is given by the LP relaxation of $\psi^p$. 
          $\hfill \lozenge$
   \end{definition}       
          As far as we know, the fractional version of $\psi^p$ has not been studied before. It is possible to show that $\beta({\cal H}) \leq \psi_{f}^p \leq \psi_f$ (simple extension to arguments in \cite{tehrani2012bipartite}). 
   
    	In our prior work \cite{Localpap1}, for the UIC problem on a side information digraph $G_d$,  we have shown that there are index coding achievable schemes based on local and fractional local chromatic numbers defined on the interference graph $\bar{G}_d$, denoted by $\chi_{\ell} \left( \bar{G}_d \right)$ and $\chi_{f \ell} \left( \bar{G}_d \right)$ respectively. Now, we define the GIC analogues of $\chi_{\ell} \left( \bar{G}_d \right)$ and its fractional version $\chi_{f \ell} \left( \bar{G}_d \right)$. As far as we are aware, we have not encountered these generalizations for the GIC problem on directed bipartite graphs. 
           
    \begin{definition}
        The \textit{local hyperclique cover} of ${\cal H}$, denoted $\psi_\ell \left( {\cal H} \right)$, is given by the following integer program:
                       \begin{align}\label{prog:lochyp}
                                     \hfill & \min~ t \nonumber \\
                                      \mathrm{s.t.}& \sum \limits_{ C: W \left( R(u) \bigcup (S(u))^c \right) \bigcap C \neq \emptyset} y_C    \leq t ,~\forall u \in U \nonumber \\
                                      \hfill & \sum \limits_{C: u \in C} y_C = 1, ~\forall u \in U \nonumber\\
                                        \hfill & y_C \in \{0,1\}~\forall C \in {\cal C},~ t \in \mathbb{Z}^{+}
                       \end {align}   
              where ${\cal C}$ is the set of hypercliques in ${\cal H}$.          $\hfill \lozenge$
    \end{definition} 
       
       The LP relaxation of (\ref{prog:lochyp}) is defined to be the \textit{fractional local hyperclique cover}, denoted $\psi_{f \ell} \left( {\cal H} \right)$. Note that, the UIC analogues of $\psi_\ell$ and $\psi_{f \ell}$ are $\chi_\ell \left(\bar{G}_d \right)$ and $\chi_{f \ell} (\bar{G}_d)$ \cite{Localpap1} respectively.  Now, we provide a brief description about the feasible solution to (\ref{prog:lochyp}). For a user $u$, let us call the set of users that request packets not in  $S_u$ to be the \textit{interference neighborhood}. The interference neighborhood consists of: 1) users requesting the same packet as the user ($R(u)$). 2) users requesting packet neither in $S_u$ nor $R(u)$. For any user $u$, given the feasible hyperclique cover, we count the number of hypercliques, belonging to the cover, in user $u$'s interference neighborhood.  Let us call this \textit{local hyperclique count} of user $u$. $t$ denotes the maximum local hyperclique counts over all users.  Then finally minimizing $t$ over all possible hyperclique covers, gives $\psi_\ell$. In this work we will show that $\psi_\ell$ is an upper bound to $\beta$.

  We define a new achievable scheme for the GIC problem by combining ideas from local hyperclique cover and partition multicast. This new scheme is called \textit{partitioned local hyperclique cover} denoted by $\psi_{\ell}^p$. Now, we briefly discuss the motivation behind defining $\psi_{\ell}^p$. 
  
  For simplicity, let us consider the UIC problem on directed side information graphs. Recall that $\chi_f \left( \bar{G}_d\right)$ is the optimal way of fractionally covering a digraph $G_d$ with cliques. Since, a subset of a clique is a clique, partitioning a graph into different groups and then adding up the clique covers of each group is not going to be better than covering the whole graph with cliques without partitioning. 
       \begin{figure}
       \centering
           \includegraphics [width=9cm]{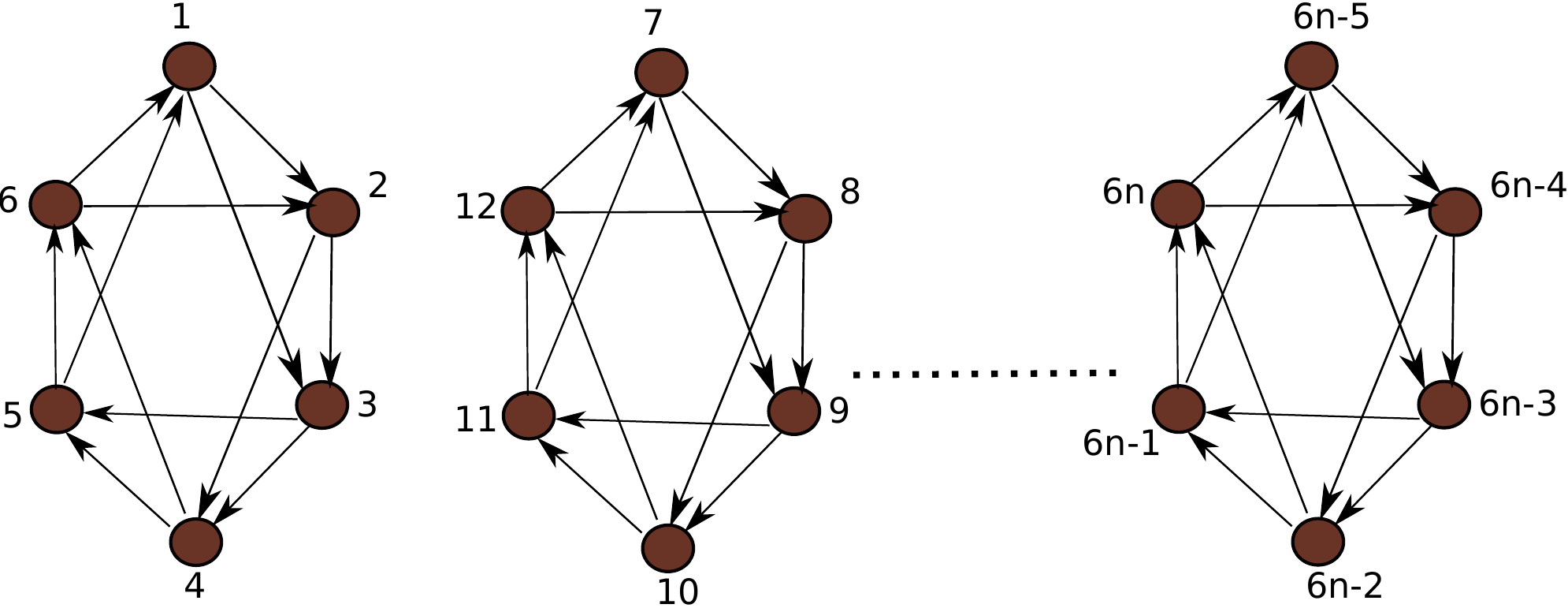}
        \caption{An example UIC problem with a side information graph $G_d$ for which $\chi_f \left(\bar{G}_d \right) = 6n$. The partition multicast number and its fractional versions are both $4n$. Partitioning into component $6$-vertex graphs and adding up the local chromatic numbers of their complements gives $4n$, i.e. $\chi_{f \ell}^p =4n$. }   
        \label{Fig:uniongraph}
       \end{figure}
       
       However, even for directed graphs, partitioning may help when it comes to $\chi_{\ell} \left(\bar{G}_d \right)$. An example illustrating this is given in Fig. \ref{Fig:uniongraph}. The directed side information graph $G_d$, given in Fig. \ref{Fig:uniongraph}, is a union of $n$ different $6$-vertex graphs. In each $6$-vertex graph, every vertex has an out edge to the next two vertices in the ordering. Observe that, there is no clique of size greater than $1$. Hence, the optimal clique cover is obtained by assigning every vertex to a different clique. Therefore, $\beta \leq \chi_f \left( \bar{G}_d \right) =6n$ for the digraph in Fig. \ref{Fig:uniongraph}. In addition, the local chromatic number involves counting the number of cliques in the complement of the neighborhood of any vertex. This gives $\beta \leq \chi_{\ell} \left( \bar{G}_d \right)= \chi_{f \ell} \left( \bar{G}_d \right)=6n-2$. Now,  considering each 6 vertex graph individually,  computing its local chromatic number, and summing up results in the bound $\beta \leq 4n$. Therefore, partitioning provides a significant improvement. This operation of partitioning ${\cal H}$, computing local  hyperclique covers for each subgraph, and adding them up is captured by the definition of $\psi_{\ell}^p$. Incidentally, the partition multicast number for this case is also $4n$. Now, we directly define the GIC variant the combines the idea from partition multicast and local  hyperclique cover.    
                        
   \begin{definition}
             The \textit{partitioned local hyperclique cover number} of ${\cal H}$, denoted $\psi_{\ell}^{p} \left( {\cal H} \right)$, is given by the following integer program:
                    \begin{align}
                            \hfill & \min~ \sum \limits_{M} a_M t_M \nonumber \\
                              \mathrm{s.t.}~ & \sum \limits_{ C: W\left( R(u) \bigcup (S(u))^c \right) \bigcap C \bigcap M \neq \emptyset} y_C  \leq t_M ,~\forall u \in M,~\forall M \in 2^U \nonumber\\ 
                               \hfill & \sum \limits_{M:v \in M} a_M =1, \forall u \in U \nonumber\\
                              \hfill & \sum \limits_{C:v \in C} y_C =1, \forall u \in U \nonumber \\
                              \hfill & a_M,y_C \in \{0,1\} ~\forall M \in 2^{U}-\{\emptyset \}, ~C \in {\cal C}, ~t_M \in \mathbb{Z}^{+} 
                              \label{prog:plochyp}
                     \end{align}  
               where ${\cal C}$  is the set of hypercliques in ${\cal H}$ and $M$ is a multicast group.   $\hfill \lozenge$
    \end{definition}
    
     The fractional version of $\psi_{\ell}^p \left( {\cal H} \right)$, denoted by $\psi_{f \ell}^p \left( {\cal H}\right)$ is the LP relaxation of (\ref{prog:plochyp}). Let us denote the UIC analogue of $\psi_{f \ell}^p$ by $\chi_{f \ell}^p \left(\bar{G}_d \right)$ and call it partitioned fractional local chromatic number. In a feasible solution to (\ref{prog:plochyp}), we first partition the set of users into a family of multicast groups. Separately, we cover all users using a hyperclique cover. Over all users in every group $M$, we get the maximum local hyperclique count $t_M$, restricting the interference neighborhood of every user to that group. Optimizing the sum of all such counts from different multicast groups over all possible hyperclique covers and multicast group allocations gives $\psi_\ell^p$. 
  
  For a preview on the relationships between known parameters, new parameters and new relationships between the parameters we refer the reader to Fig. \ref{Fig:Resultdiag}. In the next section, we provide achievable schemes for all parameters defined in this section. 
 
 \section{Achievable Schemes}
       We first show the existence of achievable schemes for $\psi_{\ell}$ and $\psi_{f \ell}$ . Let $\left( S(u) \right)^c = P-S(u)-R(u)$.
          \begin{theorem}\label{Thm:achievability}
                 There are achievable linear index codes corresponding to $\psi_{\ell} \left( {\cal H}\right)$ and $\psi_{f \ell} \left( {\cal H}\right)$ implying $\beta \left( {\cal H} \right) \leq \psi_{f \ell} \left( {\cal H}\right) \leq \psi_{\ell} \left( {\cal H} \right)$.
          \end{theorem}
          \begin{proof}
                   The proof is analogous to the ones for the UIC problem found in \cite{Localpap1}. The need for outlining a proof is because of additional technicalities due to the fact that user requests overlap. First, we consider the case of $\psi_{\ell} \left( {\cal H} \right)$. Consider the optimal integral solution $(t, \{ y_C\})$ to program (\ref{prog:lochyp}). Every vertex is in exactly one hyperclique which is chosen. Consider the set ${\cal C}^{\mathrm{opt}}$ of hypercliques $C$ for which $y_C=1$ (hyperclique chosen) in the optimal solution. Let $\lvert {\cal C}^{\mathrm{opt}} \rvert=s$. Consider a $t \times s$ generator $\mathbf{G}$ of an $(s,t)$ MDS code over a field $\Sigma$ of size greater than $s$. Let the $i$th column be $\mathbf{g}_i$. Assign each column to a hyperclique in ${\cal C}^{\mathrm{opt}}$. Let $C(u) \in {\cal C}^{\mathrm{opt}}$ be the unique hyperclique to which user $u$ belongs. Let packet $p$ be denoted by $x_p \in \Sigma,~ \forall p \in P$. Define an equivalence relation $\sim$ between two users $u,v \in U$ such that $u \sim v~ \mathrm{iff}~  C(u)=C(v)~\mathrm{and}~R(u)=R(v) $. Let $A_u$ be the equivalence class of $u$ under $\sim$. Note that $u \in A_u$. Let us assume that $U$ partitions into $q$ equivalence classes, i.e. $U= \bigcup \limits_{i=1}^q A_{u_q}$ . Then the transmission scheme is given by:
                   \begin{equation}\label{eq:localtrans}
                    \mathbf{y}=\sum \limits_{ i \in [q] } \mathbf{g}_{C\left( u_i \right)} x_{R \left( u_i \right)}.
                    \end{equation}
  In other words, if there are two users who request the same packet and belong to the same clique, they belong to the same equivalence class and their terms can be merged into one. The broadcast rate is given by $t$.
                   
                   We need to show that a user $u$ can decode $\mathbf{x}_{R(u)}$ from this. All the terms with $\mathbf{x}_p,~p \in S(u)$ can be cancelled due to the side information of $u$. This means that, in $\mathbf{y}$, all summands corresponding to users $W(S(u))$ do not affect decoding. Note that, $W \left( R(u) \bigcup \left( S(u) \right)^c \right) = W \left( R(u) \right) \bigcup W \left(   \left( S(u) \right)^c \right)$. 
                   
                   The first constraint in program (\ref{prog:lochyp}) ensures that the number of distinct hypercliques from ${\cal C}^{\mathrm{opt}}$ in $ W \left( R(u) \right) \bigcup W \left( \left( S(u) \right) ^c \right)$ is at most $t$. Let $A_u$ be the equivalence class to which $u$ belongs. Observe that $A_u \bigcap  W \left( \left( S(u) \right) ^c \right) =\emptyset $ because $A_u$ only has users requesting $R(u)$ and $R(u) \notin \left( S(u) \right)^c$ by definition. The number of hyperlcliques from ${\cal C}^{\mathrm{opt}}$ in $A_u$ is $1$, by definition of $A_u$. Then, $ \left( W \left( R(u) \right) - A_u \right) \bigcup W \left( \left( S(u) \right) ^c \right)$ has at most $t-1$ hyper cliques. 
                   
                   The terms corresponding to users in $ \left( W \left( R(u) \right) - A_u \right) \bigcup W \left( \left( S(u) \right) ^c \right)$ constitute the interference terms. Therefore, at most $t-1$ distinct columns from $\mathbf{G}$ interferes with $\mathbf{g}_{C(u)} $. Since any $t$ columns in $\mathbf{G}$ are linear independent, the interference can be cancelled. The difference from this and the UIC case is that user requests overlap leading to a more technical analysis.           
                     
                    Now, we move to specifying an achievable scheme for the LP relaxation  of (\ref{prog:lochyp}). Let $t,\{y_C\}_{C \in {\cal C}}$ constitute the optimum solution. Since, the constraints on the variables involve only integers, the optimal solution involves only rationals. Let $r$ denote the least common multiple of denominators of $y_C$ and $t$. Define the new variable $\hat{t}=rt$ and $\hat{y}_C= r y_{C}$. Now the new variables carry integral weights. Every hyperclique ($y_C$) is assigned an integer weight in the set $\{0,1,2, \ldots r\}$.  By the covering constraints on every vertex, every vertex is covered by exactly $r$ hypercliques. Assume that every packet $p$ is a super packet containing $r$ subpackets $\mathbf{x}_{p_1} \ldots \mathbf{x}_{p_r} \in \Sigma$. Let $\sum \limits_{C} \hat{y}_C=s$.
                     
                      If a hyperclique has weight $1 \leq q \leq r $, then consider $q$ different copies of the same hyperclique. Denote the resulting multiset of cliques by ${\cal C}^{\mathrm{opt}}$. Every hyperclique in ${\cal C}^{\mathrm{opt}}$ has weight $1$ with possible repetitions among cliques. Every user $u$ is covered by at most $r$ hypercliques from ${\cal C}^{\mathrm{opt}}$. Now assign these $r$ hypercliques to $r$ different indices of the form $(u,i), ~1 \leq i \leq r$. Hence, every index pair $(u,i)$ is assigned a subpacket $x_{\left( R(u) \right)_i}$ and a hyperclique $C(u,i)$.
                      
                       If two user requests overlap, i.e. $R(u)=R(v)$, then $\left(R(u)\right)_i=\left( R(v) \right)_i$ (all respective subpackets are identical). As before in the scalar case, define $\sim$ to be an equivalence relation such that $(u,i) \sim (v,j)  ~\mathrm{iff}~ R(u)=R(v)~ \mathrm{and}~C(u,i)=C(v,i)~\mathrm{and}~ i=j,~\forall v,u \in U~ \forall i,j \in [r]$. Let ${\cal A}_{u,i}$ be an equivalence class, of all the index pairs which denote the same subpacket as the index pair $(u,i)$ and are assigned the same hyperclique, under the relation $\sim$. Note that, $(u,i) \in {\cal A}_{u,i}$. Let the set $U \times [r]$ of all index pairs be partitioned into $b$ equivalence classes, i.e. $U \times [r] = \bigcup \limits_{i=1}^b {\cal A}_{u_{i},k_i}$ where $1 \leq k_i \leq r$ and $u_{i} \in U$. Now, consider an $\left( s,\hat{t} \right)$ MDS code over $\Sigma$ with generator $\mathbf{G}$ with columns $i$ denoted by $\mathbf{g}_i$. Assign every column to a distinct hyperclique in ${\cal C}^{\mathrm{opt}}$ such that a column is denoted by $\mathbf{g}_{C}$ after the hyperclique $C$ assigned to it. The transmission scheme is given by:
                       \begin{equation}\label{eqn:trans}
                        \mathbf{y} = \sum \limits_{i=1}^b x_{\left(R(u_{i})\right)_{k_i}} \mathbf{g}_{C(u_i,k_i)}   
                       \end{equation}
    If any two index pairs $(u,i)$ and $(v,i)$ are such that $R(u)=R(v)$ and if they are assigned the same hyperclique there is only one term corresponding to both of them.              
              
              We need to show that every user decodes all the subpackets $\mathbf{x}_{\left(R(u)\right)_i}, ~\forall i \in [r]$. We define a modified $W$ function tht produces index pairs instead of just users. Define $\tilde{W}(A) = \{ (u,i): R(u) \in A, i \in [r] \}$ for any subset $A \subseteq P$. Let us consider the decoding of subpacket $\left( R(u) \right)_{i}$.
              
                   Now, we use arguments very similar to the scalar case.  Rephrasing the first constraint in program (\ref{prog:lochyp}), the number of distinct hypercliques in $ \tilde{W} \left( R(u) \right) \bigcup \tilde{W} \left( \left( S(u) \right) ^c \right)$ is at most $\hat{t}$ different hypercliques from ${\cal C}^{\mathrm{opt}}$. Let ${\cal A}_{u,i}$ be the equivalence class to which $(u,i)$ belongs.                   
                   
                   Observe that ${\cal A}_{u,i} \bigcap  \tilde{W} \left( \left( S(u) \right) ^c \right) =\emptyset $. Hence, $ \left( \tilde{W} \left( R(u) \right) - {\cal A}_{u,i} \right) \bigcup \tilde{W} \left( \left( S(u) \right) ^c \right)$ has at most $\hat{t}-1$ different hyper cliques from ${\cal C}^{\mathrm{opt}}$. The summands in (\ref{eqn:trans}), corresponding to users in $ \left( \tilde{W} \left( R(u) \right) - {\cal A}_{u,i} \right) \bigcup \tilde{W} \left( \left( S(u) \right) ^c \right)$, constitute the interference terms. Therefore, at most $t-1$ distinct columns from $\mathbf{G}$ interferes with $\mathbf{g}_{C(u,i)} $. Since any $\hat{t}$ columns in $\mathbf{G}$ are linear independent, the interference can be cancelled and therefore user $u$ can decode $\mathbf{x}_{\left( R(u)\right)_i}$. Note that, the terms involving $\mathbf{x}_{\left( R(u) \right)_j}$ constitute interference when $j \neq i$.
 Since, every user receives $r$ subpackets and the total number of transmissions is $\hat{t}$, rate is given by $t$. This concludes the proof.                    
                   
%                   For the LP relaxation (fractional version) of (\ref{prog:lochyp}), very analogous to Lemma \ref{lem:blowup}, it can be argued that $\psi_{f \ell} ({\cal H}) = \inf \limits_r \frac{\psi_{\ell} \left( {\cal H}^r \right)} {r} $. Also, the optimum is attained for a finite $r$. Given $r$, every user in ${\cal H}$ is split into $r$ subusers in the blowup  ${\cal H}^r$. Consider the optimum scalar code achieving $\psi_{\ell} \left( {\cal H}^r \right)$. To convert this into a code for the original problem, it is enough to club the packets desired by $r$ users as a super packet desired by the original user before the blowup. It is as though every original user requires $r$ sub packets. But side information possessed by a sub user in the blowup graph ${\cal H}^r$ is a subset of the side information possessed by the original user  in terms of these sub packets. This ensures that all the decoding requirements in the original ${\cal H}$ are satisfied. Since the subpacketization is $r$, being a vector linear solution on ${\cal H}$, the rate is normalized by $r$. 
          \end{proof}
    
  Now, we show that achievable schemes exist for all parameters that are based on partition multicast.  
   \begin{theorem}
              For a GIC on ${\cal H}$, there exist achievable index coding schemes whose broadcast rates equal $\psi_f^p, \psi^p_{\ell}$ and $\psi^p_{f \ell}$.
    \end{theorem}
    \begin{proof}
              We begin with the proof for the fractional partition multicast number, denoted by $\psi_f^p$, given by the LP relaxation of program (\ref{pr:partmul}). Let us first consider the integer version given by (\ref{pr:partmul}) before moving onto the LP relaxation. In an optimal solution, the set of users is partitioned into multicast groups, i.e. $U = \bigcup_{i=1}^k M_i$. Every user belongs to one multicast group. In a multicast group, the minimum size of the side information set is found. It is given by $  \min \limits_{u \in M} \lvert R(M) \bigcap S(M) \rvert$. Note that, once a multicast group is considered, the problem is to satisfy only the users in the multicast group and only their packets participate in the transmission. Hence, the relevant induced bipartite graph is ${\cal H} \left(M, R(M) \right)$. 
  
                 The transmission scheme is given by an $(R(M),d_M)$ MDS code. Since, every user has $R(M)-d_M$ number of distinct packets as side information, by the MDS property every user in the multicast can decode his request. The overall scheme is given by time sharing the different multicast groups, i.e. $\{M_1,M_2 \ldots M_k \}$ in the partition. 
                 
                 For the LP relaxation, every nonzero subset $M \in 2^U$ is a multicast group. The transmission scheme for each group is the same as the scalar case. The only difference is in time sharing. Since, the program in (\ref{pr:partmul}) has only integer coefficients, the real optimal solutions $a_M$ are rational. As in other proofs, let $r$ be the least common multiple of denominators of $a_M$. Let $\hat{a}_M=r a_M$. With the new variables, the first constraint in the LP relaxation of (\ref{pr:partmul}) implies that every user is in exactly $r$ multicast groups. Hence, every user packet consists of $r$ subpackets and each subpacket is transmitted using the scalar scheme corresponding to one of the $r$ multicast groups. Hence, $\psi_f^p$ is achievable because of rate normalization by $r$.
                 
                 Now, we provide an achievable scheme for $\psi_{\ell}^p$ (integer program (\ref{prog:plochyp})). Given a multicast group $M$, the variables $t_M, \{y_C\}$ constitute a scalar achievable scheme with $t_M$ transmissions identical to the one used to achieve $\psi_{\ell}$, as in Theorem \ref{Thm:achievability}, for the GIC problem (defined on ${\cal H} (M,R(M))$) induced by the multicast group. And the various disjoint multicast groups are timeshared. This provides an achievable scheme for $\psi_{\ell}^p$.
                 
                 Now, consider $\psi_{f \ell}^p$ (LP relaxation of (\ref{prog:plochyp}) ). Let the optimal solution be given by real values $\{t_M\},\{y_C\}, \{ a_M\}$. Since, the program has only integral constraints, all variables are rational.  Now, consider $r_1$ to the least common multiple of denominators of $\{y_C\},\{ t_M\}$. Now, define $\hat{y}_C=r_1 y_C$ and $\hat{t}_M =r_1 t_M$. All the new variables are integral. For a particular multicast group $M$, apply the vector coding scheme of $\psi_{f \ell}$ on the GIC problem (defined on ${\cal H} \left(M,R(M)\right)$) induced by the group $M$. This needs $r_1 t_M$ transmissions and every user in $M$ gets $r_1$ subpackets. Call this scheme ${\cal S}_M$. Now, let $r_2$ be the least common multiple of $\{a_M \}$. Now, consider $r_1 r_2$ sub packets for every user. Now use the scheme $S_M$, which transmits $r_1$ subpackets, $a_M r_2$ times. Since, every user is exactly in $r_2$ subgroups (by constraint $2$ in the LP relaxation of (\ref{prog:plochyp})), every user gets $r_1 r_2$ subpackets. The total number of subpackets transmitted is  $\sum \limits_{M} a_M r_2 (t_M r_1) $. Dividing by $r_1r_2$, we get the same broadcast rate as the objective in (\ref{prog:plochyp}).   
                               
    \end{proof}

 \section{Relationship between different parameters}                   
         In this section, we provide bounds for ratios between different parameters. In our prior work \cite{Localpap1}, we showed that $\chi_f/ \chi_{f \ell} \leq \frac{5}{4}e^2$ for all UIC problems. A parallel and independent work \cite{simonyi2013} has shown a tighter upper bound of $e$. This means that the performance of the achievable scheme due to $\chi_{f \ell}$ is at most $e$ away from the one based on $\chi_{f}$ for the UIC problem. But it was not clear how to generalize $\chi_{f \ell}$ to the GIC problem and also what relationship such a possible generalization would have with respect to $\psi_f$. In this work, we have defined the GIC counterpart for $\chi_{f \ell}$.  Now, we show that the generalizations based on local chromatic numbers to the GIC problem satisfy similar bounds in relation to the generalization of $\chi_f(\bar{G}_d)$, i.e. $\psi_f$.
         
        \begin{theorem}\label{thm:bound}
                       $\frac{\psi_f \left( {\cal H} \right) }{ \psi_\ell \left( {\cal H} \right)} \leq \frac{\psi_f \left( {\cal H} \right) }{ \psi_{f \ell} \left( {\cal H} \right)} \leq e$.
          \end{theorem}
          \begin{proof}
                The left inequality is obvious because $\psi_{f \ell}$ is the LP relaxation of $\psi_{\ell}$. To prove the right inequality, given any GIC problem on ${\cal H}$, we come up with a UIC problem on a side information graph $G_d$ such that $\psi_f \left({\cal H}\right) = \chi_f \left( \bar{G}_d\right)$ and $\psi_{f \ell} \left({\cal H}\right) \geq  \chi_{f \ell} \left( \bar{G}_d\right) $.  This will imply that $\frac{\psi_f \left( {\cal H} \right) }{ \psi_{f \ell} \left( {\cal H} \right)} \leq \frac{ \chi_f (\bar{G}_d)}{ \chi_{f \ell} (\bar{G}_d)}$. Since $ \frac{ \chi_f (\bar{G}_d)}{ \chi_{f \ell} (\bar{G}_d)}$ this has been shown \cite{simonyi2013} to be upper bounded by $e$  for all digraphs $G_d$, this implies the right inequality.
                
                In the GIC problem on ${\cal H}$, the number of packets is less than that of the number of users. To convert it into a UIC problem with side information graph $G_d$ on the user set $U$, we introduce a packet for every user such that user $u$ requests packet $x_u$. The user set is identical to both problems. For user $u$, if $v \in W(S(u))$ in the GIC problem, then $(u,v) \in G_d$. In other words, if user $u$ has a packet requested by user $v$ in GIC problem, packet $x_v$ is present as side information with user $u$ in the UIC version. Further, if requests of users $v$ and $u$  are identical in the GIC problem, then $(u,v), (v,u) \in E_d$. 
                
               From the above construction, it is clear that if $C$ is a hyperclique in ${\cal H}$, then $C$ is a clique in the side information graph $G_d$ and vice versa. A clique $C$ in a directed graph is defined to be the complete graph on the vertices of $C$. Any two vertices in a complete graph have edges in both directions.  $\psi_f$ is an efficient covering of all users by hypercliques. $\chi_f \left( \bar{G}_d \right) $ is also an efficient covering of all users in $G_d$ by cliques in $G_d$. Since, the user set is identical and the set of cliques is identical to the set of hypercliques, $\psi_f ({\cal H})= \chi_f \left( \bar{G}_d \right)$.
               
               Now, we show that $\chi_{f \ell} (G_d) \leq \psi_{f \ell} \left( {\cal H} \right)$.  Consider a hyperclique $C$ in ${\cal H}$. Then $C$ is a clique in $G_d$. Consider some arbitrary weights $y_C$ for all $C$ and consider a scalar $t$ that satisfies the constraints in LP relaxation of program (\ref{prog:lochyp}). Since, the user set is identical if the weights $y_C$ assigned cover every user $u$ in ${\cal H}$, the covering constraint holds for $G_d$ too. Let $N(u)$ represent the out-neighborhood of $u$ in $G_d$. Let $\left( N(u) \right)^c = U-N(u)-{u}$. 
               
               The equivalent of constraint $1$ of the LP relaxation of program (\ref{prog:lochyp}) for $\chi_{f \ell} \left( \bar{G}_d \right)$ is: $\sum \limits_{C \bigcap \left( \left(N(u) \right)^c \bigcup u \right)} y_C \leq t$.   
               
               It is enough to show that this equivalent constraint is satisfied for $t$. Observe that $ \left( N(u) \right)^c  \subseteq W\left( \left( S(u) \right)^c \right)  $ and $u \in W(R(u))$. Therefore, $ C \bigcap \left( \left(N(u) \right)^c \bigcup u \right) \subseteq  C \bigcap W \left( R(u) \right) \bigcup W \left( \left( S(u) \right) ^c \right) $. Therefore, the equivalent constraint holds.
               
               From the above two results, we obtain $\frac{\psi_f \left( {\cal H} \right) }{ \psi_{f \ell} \left( {\cal H} \right)} \leq \frac{\chi_f \left( \bar{G}_d \right) } { \chi_{\ell} \left( \bar{G}_d \right)} $.
   The ratio on the right has been shown to be upper bounded by $e$ as noted before. Hence, the result in the theorem follows. 
                       
          \end{proof}
                                              
       Now, we show that $\psi_{f \ell}^{p}$ is better than all the graph theoretic schemes.
  
         \begin{theorem}
                 The achievable scheme based on $\psi_{f \ell}^p$ is better than all known previous achievable schemes based on the concepts of hyperclique covers, local graph coloring and partitioning. Formally, $\psi_{f \ell}^p \leq \psi_{f \ell} \leq \psi_f$ and $\psi_{f \ell}^p \leq \psi_f^p$.
         \end{theorem}
        \begin{proof}
                Setting $a_U=1$ and $a_M=0, ~\forall M: \lvert M \rvert < \lvert U \rvert $ in the LP relaxation of program (\ref{prog:plochyp}) and optimizing for the rest of the variables one gets $\psi_{f \ell}$ which is less than or equal to $\psi_f$ by definition.  Hence, the first chain of inequalities is proved. 
                
                 Now, we show that the fractional partition multicast number $\psi_f^p \geq \psi_{f \ell}^p $. To see this, consider the optimal solution of the relaxation of program (\ref{pr:partmul}) given by $a_M$. Now, it is enough to show that there exists a set of feasible variables $y_C$ and $t_M=d_M$ such that the constraints of the LP relaxation of program (\ref{prog:plochyp}) are satisfied. If $C= W (p)$ for some packet $p \in P$, then $y_C=1$, otherwise $y_C=0$. In other words, assign a weight of $1$ to those hypercliques that comprise the set of users requesting the same packet $p$ for some $p \in P$ and all other hypercliques are assigned a weight $0$.  Since, every user is contained in a unique hyperclique characterized by the packet the user requests, constraint $3$ in the LP relaxation of program (\ref{prog:plochyp}) is satisfied. Constraint $2$ is satisfied because the set of variables $a_M$ form a feasible solution to the LP relaxation of (\ref{pr:partmul}).  We need to show that the first constraint is satisfied with $t_M=d_M$ for all $M$ and $u \in M$. Consider a particular $M$. By the assignment of variables $y_C$, we have the following chain of inequalities:             
                 \begin{equation}
                      \sum \limits_{ C: W\left( R(u) \bigcup (S(u))^c \right) \bigcap C \bigcap M \neq \emptyset} y_C = R(M) - \lvert R(M)- S(u) \rvert \leq d_M 
                 \end{equation} 
     This is because, the assignment of values imply that a hyperclique with non-zero weight  is 'synonymous' with the packet and the number of hypercliques in $W\left( R(u) \bigcup (S(u))^c \right) \bigcap  M$ is exactly  the number of packets $R(M)-S(u)$. This completes the proof.           
        \end{proof}
        
       Now, we state the final result of the paper. This implies that rates of all the achievable schemes discussed in this work are at most a factor $e$ far away from $\psi_f \left( {\cal H} \right)$.
          \begin{theorem}
                 $ \psi_f \leq e \psi_{f \ell}^p $.
         \end{theorem}
          \begin{proof}
                   Let the set of variables $\{t_M,y_C,a_M\}$ be the optimal solution to the LP relaxation of program (\ref{prog:plochyp}). Let ${\cal C}_M$ be the set of hypercliques for the induced GIC problem on ${\cal H} \left(M,R(M) \right)$. Let $\hat{y}_{C_M} =\sum \limits_{C: C \bigcap M = C_M} y_C, ~ \forall C_M \in {\cal C}_M$. Observe that, since variables $\{t_M\}$ and $\{y_C\}$ satisfy the first constraint of the LP relaxation of program (\ref{prog:plochyp}), the variables $\hat{y}_C$ and the variable $\{ t_M \}$ satisfy the constraints of the LP relaxation of program (\ref{prog:lochyp}) on the induced GIC problem given by ${\cal H} \left(M, R(M) \right)$. Therefore, $t_M \geq \psi_{f \ell} \left( {\cal H} \left( M, R(M) \right)  \right)$. Now, we have the following chain of inequalities:
                  \begin{align}      
           \frac{\psi_f \left( {\cal H} \right) }{ \psi_{f \ell}^p \left( {\cal H} \right)} & \leq \frac{\psi_f \left({\cal H} \right)}{\sum \limits_{M \in 2^U-\{0\} } a_M t_M} \nonumber\\
              \hfill                                                                                               & \overset{a}{\leq}  \frac{\psi_f \left({\cal H} \right)}{\sum \limits_{M \in 2^U-\{0\} } a_M \psi_{f \ell} \left( {\cal H} \left( M, R(M) \right)  \right)} \nonumber \\
              \hfill                                                                                               & \overset{b}{\leq} \frac{ \sum \limits_{M \in 2^U-\{0\} } a_M \psi_{f} \left( {\cal H} \left( M, R(M) \right)  \right)} {\sum \limits_{M \in 2^U-\{0\} } a_M \psi_{f \ell} \left( {\cal H} \left( M, R(M) \right)  \right)}    \nonumber \\      
              \hfill                                                                                               & \overset{c}{\leq} \frac{ \sum \limits_{M \in 2^U-\{0\} } a_M \psi_{f} \left( {\cal H} \left( M, R(M) \right)  \right)} {\sum \limits_{M \in 2^U-\{0\} } a_M \frac{1}{e} \psi_{f } \left( {\cal H} \left( M, R(M) \right)  \right)} \nonumber \\
              \hfill                                                                                               & \leq e               
                  \end{align}  
      Justifications for the above chain are: a) $t_M \geq \psi_{f \ell} \left( {\cal H} \left( M, R(M) \right)  \right) $ b) Partitioning the set of users and adding up $\psi_f$ over all partitions can not increase $\psi_f$ for the original GIC problem. c) Theorem \ref{thm:bound}.            
           \end{proof}
      \section{Conclusion}    
             In this work, we generalized the concept of local and fractional local chromatic numbers and their achievable schemes to the groupcast setting. Further, we defined a new graph parameter and an achievable scheme combining local coloring concepts and the idea of partitioning. This scheme is better than all known purely graph theoretic schemes and generalizes clique covers, cycle covers, partition multicast and local coloring. We show that this scheme is multiplicatively at most $e$ far away from the scheme based on fractional chromatic number (or fractional hyperclique cover number).
\pagenumbering{arabic}
\bibliographystyle{IEEEtran}
\bibliography{Local2bib}

\end{document}